\documentclass[12pt,journal]{IEEEtran}

\usepackage[hmargin=1.2in,vmargin=1.2in]{geometry}
\usepackage{amsmath,amsthm,amssymb,amsfonts,verbatim}
\usepackage{graphicx}
\usepackage[colorlinks, linkcolor=blue,  anchorcolor=blue,
            citecolor=blue
            ]{hyperref}
\usepackage[hmargin=1.2in,vmargin=1.2in]{geometry}
\usepackage{booktabs}

\title{Binary sequences with low correlation via cyclotomic function fields}
\author{Lingfei Jin, Liming Ma and Chaoping Xing}

\newtheorem{lemma}{Lemma}[section]
\newtheorem{theorem}[lemma]{Theorem}

\newtheorem{prop}[lemma]{Proposition}

\theoremstyle{remark}

\renewcommand{\epsilon}{\varepsilon}
\renewcommand{\le}{\leqslant}
\renewcommand{\ge}{\geqslant}

\def\PGL{{\rm PGL}}
\def\Gal{{\rm Gal}}
\def\bal{{\rm bal}}
\newcommand{\vnote}[1]{}

\def\PP{\mathbb{P}}

\def\F{\mathbb{F}}
\def\Z{\mathbb{Z}}

\def \mL {\mathcal{L}}

\def \mS {\mathcal{S}}

\def \Xi {{X^{[i]}}}

\newcommand{\Ga}{\alpha}

\newcommand{\Gd}{\delta}

\newcommand{\Gl}{\lambda}    
     
\newcommand{\Gs}{\sigma}

\newcommand{\s}{\sigma}

\def\PGL{{\rm PGL}}

\def \bs {{\bf s}}

\def \bu {{\bf u}}
\def \bv {{\bf v}}

\def\mS{{\mathcal S}}
\def\Aut {{\rm Aut }}

\def\Gal{{\rm Gal}}

\onecolumn

\begin{document}
\maketitle

\begin{abstract}
Due to wide applications of binary sequences with low correlation to communications, various constructions of such sequences have been proposed in literature.
However, most of the known constructions via finite fields make use of the multiplicative cyclic group of $\F_{2^n}$. It is often overlooked in this community that all $2^n+1$ rational places (including ``place at infinity") of the rational function field over $\F_{2^n}$ form a cyclic structure under an automorphism of order $2^n+1$. In this paper, we make use of this cyclic structure to provide an explicit construction of families of binary sequences of length $2^n+1$ via the finite field $\F_{2^n}$. Each family of sequences has size $2^n-1$ and its correlation is upper bounded by $\lfloor 2^{(n+2)/2}\rfloor$. Our sequences can be constructed explicitly and have competitive parameters. In particular, compared with the Gold sequences of length $2^n-1$ for even $n$, we have larger length and smaller correlation although the family size of our sequences is slightly smaller.
\end{abstract}

\section{Introduction}
Sequences with good correlation properties have been widely used in many applications, such as code-division multiple access (CDMA), spread spectrum systems and broadband satellite communications. For such applications, it is desirable to construct sequences with low autocorrelation, low cross-correlation, and large family size.

In literature, there are various constructions of binary sequences with good parameters. The best known class of such sequences are the so called m-sequences, which can be obtained from linear feedback shift registers \cite{m}, with length of the form $2^n-1$ (for a positive integer $n$).
Since m-sequences have ideal autocorrelation properties,  many families of low correlation sequences have been constructed by using m-sequences and their decimations.
For example, the well-known Gold sequences are a family of binary sequences having three-level out-of-phase auto- and cross-correlation (nontrivial correlation) values  which can be obtained by EXOR-ing two m-sequences of the same length with each other \cite{Span}.
Kasami sequences can be constructed using m-sequences and their decimations \cite{Kas}.
Besides, there are some known families of sequences of length  $2^n-1$ with good correlation properties, such as  bent function sequences \cite{Bent}, No sequences \cite{No}, TN sequences \cite{Klap}.
In 2011, Zhou and Tang generalized the modified Gold sequences and obtained binary sequences with length $2^n-1$, $n=2m+1$, size $2^{\ell n}+\cdots+2^n-1$, correlation $2^{m+\ell}+1$ \cite{ZT}.
In addition to the above mentioned sequences of length  $2^n-1$, many efforts have been devoted to the constructions of binary sequences with length of other forms.
In \cite{Uda}, Udaya and Siddiqi obtained a family of $2^{n-1}$ binary sequences, length $2(2^n-1)$  satisfying the Welch bound for odd $n$.
In \cite{Tang3}, this family was extended to a larger one with  $2^n$ sequences and the same correlation. In \cite{Tang2}, the authors presented two optimal families of sequences of length $2^n-1$ and $2(2^n-1)$ for odd integer $n$.
In 2002, Gong presented a construction of binary sequences with size $2^n-1$, length $(2^n-1)^2$, and correlation $3+2(2^n-1)$ \cite{gong}.
All of these constructions are based on finite fields of even characteristic.


There are few constructions of binary sequences which are based on finite fields of odd characteristics \cite{Lempel,Legendre,Pat,Rushanan}.
In \cite{Pat}, Paterson proposed a new family of pseudorandom binary
sequences based on Hadamard difference sets and MDS codes for length $p^2$ where $p\equiv3 (\text{mod } 4)$ is a prime.
In 2006, a family of binary sequences, called Weil sequences, which have prime length were designed  \cite{Rushanan}. The idea  is to derive sequences from single quadratic residue based on Legendre sequences using a shift-and-add construction. This sequence family has sequence length $p$ ($p$ is an odd prime), family size $(p-1)/2$, and correlation bounded by $5+2\sqrt{p}$.
In 2010, Su et.al. introduced a new sequence of period $p(q-1)$ by combining the $p$-periodic Legendre sequence and the $(q-1)$-periodic Sidelnikov sequence, called the Lengdre-Sidelnikov sequence \cite{Su}.
Recently, Jin et. al. considered the constructions of binary sequences with low correlation via multiplicative quadratic character over finite fields of odd characteristics \cite{Jin21}. They can produce families of good sequences with length $q-1$ for any prime power $q$ by using cyclic multiplication group $\F_q^*$ and sequences of  length $q$ for any odd prime $q$  from cyclic additive group $\F_q$.

Though  a lot of work have been done on the constructions of binary sequences with good parameters, there are still limited choices for the sequence length in the known results. Thus, it is of great importance to design good sequences with other possible lengths
for different scenarios, since adding or deleting sequence values, usually destroys the  good correlation
properties in these families.

\subsection{Our main result and techniques}

In this paper, we provide an explicit construction of families of binary sequences of length $2^n+1$ via the finite field $\F_q$ with $q=2^n$ elements. The correlation of this family is upper bounded by $\lfloor2^{(n+2)/2}\rfloor$. To our knowledge, this is the first construction of binary sequence of length $2^n+1$ with low correlation. In Table I, we  list some well-known families of binary sequences  with low correlation for comparison. Our main result is listed in the last row of Table I.

Most of the constructions of binary sequences with low correlation in literature explore cyclic structure of a finite field. For instance, for $q=2^n$, one can use the multiplicative cyclic group $\F_{2^n}^*$. This structure  produces sequences of length $2^n-1$. All binary sequences in Table I make use of this multiplicative cyclic structure or modify  the constructions based on this multiplicative cyclic structure. One of the fact that all points of the projective line over $\F_q$ (together with ``the point at infinity") form a cyclic group is often overlooked. In this paper, we make use of this fact and provide an explicit construction of binary sequences of length $q+1$, size $q-1$ with low correlation via finite fields $\F_q$ where $q=2^n$ and $n$ is a positive integer. It turns out that our sequences have competitive parameters. For instance, compared with the Gold sequences of length $2^n-1$ for even $n$, we have larger length and smaller correlation although the family size of our sequences is slightly smaller.

From a high level point of view, we may view the zero $P_\Ga$ of $x-\Ga$ as a point of the projective line over $\F_q$ for any $\Ga\in\F_q$.
In order to obtain sequences with length $q+1$, we can include the infinity place of the rational function field $\F_q(x)$ as \cite{JMX20,JMX21}.
It is well known that the automorphism group of the rational function field $\F_q(x)$ over $\F_q$ is isomorphic to the projective general linear group $\PGL_2(\F_q)$ and  there is an automorphism $\Gs$ of order $q+1$ such that $\{\Gs^j(P_0)\}_{j=0}^q$ form all $q+1$ rational places of $\F_q(x)$ \cite{HKT08,JMX20}.
Our sequences are constructed via the absolute trace of evaluations of rational functions on such cyclically ordered rational places $\{\Gs^j(P_0)\}_{j=0}^q$ from $\F_q$ to $\F_2$.
To obtain a lower correlation, we need to choose rational functions $z_1,z_2,\cdots,z_N$ for some positive integer $N$ carefully.
The correlation is converted to estimating the bound of the number of rational places of Artin-Schreier curves $y^2+y=z_i+\s^{-t}(z_j)$ for some $0\le i, j\le N-1$ and $0\le t\le q$.
Firstly,  $z_i+\Gs^{-t}(z_i)$ can not be constants in $\F_q$ for $1\le t\le q$ and $1\le i\le N$. To overcome this problem, we introduce an equivalence relation and choose at most one element in each equivalence class. Secondly, in order to obtain an Artin-Schreier curve with small genus, we need to choose function $z_i$ such that $z_i$ and $\Gs^{-t}(z_j)$ have the same pole. In particular, $z_i$ can be chosen from a Riemann-Roch space associated to a place which is invariant under the automorphism $\Gs$.
Hence, we can employ the cyclotomic function fields over $\F_q(x)$ with modulus $p(x)$ which is a primitive quadratic irreducible polynomial to construct such explicit binary sequences with low correlation.

\begin{table}[]
	\setlength{\abovecaptionskip}{0pt}%
	\setlength{\belowcaptionskip}{10pt}%
	\caption{PARAMETERS OF SEQUENCE FAMILIES}
	\centering

	\begin{tabular}{@{}cccc@{}}
		\toprule
		Sequence                             & Sequence Length N                                & Family Size                            & Max Correlation                       \\ \midrule
		\multicolumn{1}{|c|}{Gold(odd)}      & \multicolumn{1}{c|}{$2^n-1$, $n$ odd}                   & \multicolumn{1}{c|}{$N+2$}             & \multicolumn{1}{c|}{$1+\sqrt{2}\sqrt{N+1}$} \\ \midrule
		\multicolumn{1}{|c|}{Gold(even)}     & \multicolumn{1}{c|}{$2^n-1,n=4k+2$}            & \multicolumn{1}{c|}{$N+2$}             & \multicolumn{1}{c|}{$1+2\sqrt{N+1}$}        \\ \midrule
		\multicolumn{1}{|c|}{Kasami(small)}  & \multicolumn{1}{c|}{$2^n-1, n \text{ even}$}            & \multicolumn{1}{c|}{$\sqrt{N+1}$}      & \multicolumn{1}{c|}{$1+\sqrt{N+1}$}         \\ \midrule
		\multicolumn{1}{|c|}{Kasami(large)}  & \multicolumn{1}{c|}{$2^n-1,n=4k+2$}            & \multicolumn{1}{c|}{$(N+2)\sqrt{N+1}$} & \multicolumn{1}{c|}{$1+2\sqrt{N+1}$}        \\ \midrule
		\multicolumn{1}{|c|}{Bent}           & \multicolumn{1}{c|}{$2^n-1, n=4k$}                & \multicolumn{1}{c|}{$\sqrt{N+1}$}         & \multicolumn{1}{c|}{$1+\sqrt{N+1}$}          \\ \midrule
		\multicolumn{1}{|c|}{No}           & \multicolumn{1}{c|}{$2^n-1, n \text{ even}$}                & \multicolumn{1}{c|}{$\sqrt{N+1}$}         & \multicolumn{1}{c|}{$1+\sqrt{N+1}$}          \\ \midrule
		\multicolumn{1}{|c|}{TN}           & \multicolumn{1}{c|}{$2^n-1, n=2km$}                & \multicolumn{1}{c|}{$\sqrt{N+1}$}         & \multicolumn{1}{c|}{$1+\sqrt{N+1}$}          \\ \midrule
		\multicolumn{1}{|c|}{Tang et al.}  & \multicolumn{1}{c|}{$2(2^n-1),n=em$ odd }            & \multicolumn{1}{c|}{$\frac{N}{2}+1$} & \multicolumn{1}{c|}{$2+\sqrt{N+2}$}        \\ \midrule
		\multicolumn{1}{|c|}{Gong}  & \multicolumn{1}{c|}{$(2^n-1)^2$, $2^n-1$ prime}            & \multicolumn{1}{c|}{$\sqrt{N}$ }& \multicolumn{1}{c|}{$3+2\sqrt{N}$}        \\ \midrule

		\multicolumn{1}{|c|}{Paterson}       & \multicolumn{1}{c|}{$p^2$, p prime 3 mod 4}    & \multicolumn{1}{c|}{$N$}               & \multicolumn{1}{c|}{$5+4\sqrt{N}$}          \\ \midrule
		\multicolumn{1}{|c|}{Paterson,Gong}  & \multicolumn{1}{c|}{$p^2$, p prime 3 mod 4}    & \multicolumn{1}{c|}{$\sqrt{N}+1$}      & \multicolumn{1}{c|}{$3+2\sqrt{N}$}          \\ \midrule
		\multicolumn{1}{|c|}{Weil}           & \multicolumn{1}{c|}{$p$, odd prime}                & \multicolumn{1}{c|}{$(N-1)/2$}         & \multicolumn{1}{c|}{$5+2\sqrt{N}$}          \\ \midrule

		\multicolumn{1}{|c|}{Jin et al.} & \multicolumn{1}{c|}{$q-1$,$q\ge17$ odd prime power} & \multicolumn{1}{c|}{$N+3$}               & \multicolumn{1}{c|}{$6+2\sqrt{N+1}$}          \\ \midrule
		\multicolumn{1}{|c|}{Jin et al.} & \multicolumn{1}{c|}{$q-1$,$q\ge11$ odd prime power} & \multicolumn{1}{c|}{$N/2$}         & \multicolumn{1}{c|}{$2+2\sqrt{N+1}$}          \\ \midrule
		\multicolumn{1}{|c|}{Jin et al.} & \multicolumn{1}{c|}{$q$,$q\ge 17$ odd prime} & \multicolumn{1}{c|}{$N$}               & \multicolumn{1}{c|}{$5+2\sqrt{N}$}          \\ \midrule
		\multicolumn{1}{|c|}{Jin et al.} & \multicolumn{1}{c|}{$q$, $q \ge 11$ odd prime} & \multicolumn{1}{c|}{$(N-1)/2$}         & \multicolumn{1}{c|}{$1+2\sqrt{N}$}          \\ \midrule
		\multicolumn{1}{|c|}{{\bf Our construction}} & \multicolumn{1}{c|}{\bf {$2^n+1$}} & \multicolumn{1}{c|}{{\bf $N-2$}}         & \multicolumn{1}{c|}{\bf {$2\sqrt{N-1}$}}          \\
\bottomrule
	\end{tabular}

\end{table}

\subsection{Organization of this paper}
In Section \ref{sec:2}, we briefly introduce some backgrounds on sequences, global function fields, rational function fields, cyclotomic function fields, Kummer extensions and Artin-Schreier curves. In Section \ref{sec:3}, we provide a general construction of families of binary sequences with low correlation via cyclotomic function fields over finite fields. In Section \ref{sec:4}, an explicit construction of binary sequences is given. Besides, an algorithm to generate such a family of sequences is given correspondingly as well as some numerical results.

\section{Preliminaries}\label{sec:2}
In this section, we present some preliminaries on binary sequences and their correlation, global function fields, cyclotomic function fields, Kummer extensions and Artin-Schreier curves.

\subsection{Binary sequences and their correlation}
Let $\mS$ be a set of binary sequences, each has length $N$. For every sequence $\bs=(s_0,s_1,\dots,s_{N-1})\in\mS$ with $s_i\in\{1,-1\}$, we define the autocorrelation of $\bs$ at delay $t$ ($1\le t\le N-1$) by
\begin{equation}\label{eq:2.1}
A_t(\bs):=\left|\sum_{i=0}^{N-1}s_is_{i+t}\right|,
\end{equation}
where $i+t$ is taking modulo $N$. Consider two distinct sequences $\bu=(u_0,u_1,\dots,u_{N-1})$ and $ \bv=(v_0,v_1,\dots,v_{N-1})$ in $\mS$ with $u_i,v_j\in\{1,-1\}$, we define the cross-correlation of $\bu$ and $\bv$ at delay $t$ ($0\le t\le N-1$) by
\begin{equation}\label{eq:2.2}
C_t(\bu,\bv):=\left|\sum_{i=0}^{N-1}u_iv_{i+t}\right|.
\end{equation}
The correlation of the sequence family $\mS$ is defined by
 \begin{eqnarray}\label{eq:2.3}\nonumber
Cor(\mS):=\max\left\{\max_{\bs\in\mS, 1\le t\le N-1}\{A_t(\bs)\},\max_{\bu\neq\bv\in\mS, 0\le t\le N-1}\{C_t(\bu,\bv)\}\right\}.
\end{eqnarray}
If the sequence $\bs$ has length $N$, then the balance of $\bs$  is defined as
 \[ \bal(\bs):=\sum_{i=0}^{N-1}s_i.\]
In particular, a sequence with $\bal(\bs)=0$ is called optimally balanced if $N$ is even; a sequence with $\bal(\bs)=\pm1$ is called optimally balanced if $N$ is odd.

\subsection{Global function fields}

Let $\F_q$ denote the finite field with $q$ elements.
Let $F/\F_q$ be a function field with genus $g(F)$ over the full constant field $\F_q$, which is called a global function field. Let $\PP_F$ denote the set of places of $F$. Any place with degree one is called rational.  From the Serre bound \cite[Theorem 5.3.1]{St09}, the number $N(F)$ of rational places of $F$ is upper bounded by $$|N(F)-q-1|\le g(F) \lfloor 2\sqrt{q}\rfloor. $$

For any place $P\in \mathbb{P}_F$, let $\nu_P$ be the normalized discrete valuation of $F$ corresponding to $P$. The principal divisor of a nonzero element $z\in F$ is given by $(z):=\sum_{P\in \mathbb{P}_F}\nu_P(z)P$.
For a divisor $G$ of $F/\F_q$, the Riemann-Roch space associated to $G$ is defined by
\[\mL(G):=\{z\in F^*:\; (z)+G\ge 0\}\cup\{0\}.\]
From the Riemann-Roch theorem \cite[Theorem 1.5.17]{St09}, $\mL(G)$ is a vector space over $\F_q$ of dimension $\deg(G)-g(F)+1$ for any divisor $G$ with degree $\deg(G)\ge 2g(F)-1$.

Let $E/\F_q$ be a finite extension of function fields $F/\F_q$. The Hurwitz genus formula \cite[Theorem 3.4.13]{St09} yields
$$2{g(E)}-2=[E:F]\cdot (2g(F)-2)+\deg \text{ Diff}(E/F),$$
where $\text{Diff}(E/F)$ is the different of $E/F$.
For $P\in \PP_F$ and $Q\in \PP_E$ with $Q|P$, let $d(Q|P), e(Q|P)$ be the different exponent and ramification index of $Q|P$, respectively. Then the different of $E/F$ is given by
$$\text{Diff}(E/F)=\sum_{Q\in\PP_E} d(Q|P) Q.$$
If $Q|P$ is unramified or tamely ramified, i.e., $p\nmid e(Q|P)$, then $d(Q|P)=e(Q|P)-1$ by Dedekind's Different Theorem \cite[Theorem 3.5.1]{St09}.

We denote by  $\Aut(F/\F_q)$ the automorphism group of $F$ over $\F_q$, i.e.,
\begin{equation}
\Aut(F/\F_q)=\{\Gs: F\rightarrow F |\; \Gs  \mbox{ is an } \F_q\mbox{-automorphism of } F\}.
\end{equation}
Now we consider the group action of automorphism group $\Aut(F/\F_q)$ on the set of places of $F$.
From \cite[Lemma 1]{NX14}, we have the following results.

\begin{lemma}\label{lem:2.0}
For any automorphism $\s\in \Aut(F/\F_q)$, $P\in \mathbb{P}_F$ and $f\in F$, we have
\begin{itemize}
\item[(1)] $\deg(\s(P))=\deg(P)$;
\item[(2)] $\nu_{\s(P)}(\s(f))=\nu_P(f)$;
\item[(3)] $\s(f)(\s(P))=f(P)$ if $\nu_P(f)\ge 0$.
\end{itemize}
\end{lemma}
From the above Lemma \ref{lem:2.0}, $\sigma(\mathcal{L}(D)) = \mathcal{L}(\sigma(D))$ for any divisor $D$ of $F$ \cite{MXY16}. In particular, if $\s(P)=P$, then $\s(\mL(rP))=\mL(\s(rP))=\mL(rP)$ for any $r\in \mathbb{N}$.

\subsection{Rational function fields}
In this subsection, we introduce basic notations and facts of the rational function field. The reader may refer to \cite{JMX20,JMX21,St09} for more details.
Denote by $K$ the rational function field $\F_q(x)$, where $x$ is a transcendental element over $\F_q$.
Every finite place $P$ of $K$ corresponds to a monic irreducible polynomial $p(x)\in\F_q[x]$, and its degree $\deg(P)$ is equal to the degree of $p(x)$.
There is an infinite place of $K$ with degree one which is the pole of $x$ and denoted by $P_{\infty}$.
The set of places of $K$ is denoted by $\PP_K$. 
In fact, there are exactly $q+1$ rational places in $\PP_K$, i.e., the place $P_{x-\Ga}$ or $P_{\Ga}$ for short corresponds to $x-\Ga$ for each $\Ga\in \F_q$ and the infinite place $P_{\infty}$.

Let $P$ be a rational place of $K$ and let $\mathcal{O}_P$ be the corresponding valuation ring with respect to $P$. For $f\in \mathcal{O}_P$,  we define $f(P)\in \mathcal{O}_P/P=\F_q$ to be the residue class of $f$ modulo $P$; otherwise $f(P)=\infty$ for $f\in K\setminus \mathcal{O}_P$.
In particular, if $f(x)=g(x)/h(x)\in K$ can be written as a quotient of relatively prime polynomials $g(x)=a_nx^n+\cdots+a_0$ and $h(x)=b_mx^m+\cdots+b_0\in \F_q[x]$ with $a_n b_m\neq 0$, then the residue class map can be determined explicitly as follows 
$$f(P_\Ga)=\begin{cases} g(\Ga)/h(\Ga) & \text{ if } h(\Ga)\neq 0 \\ \infty & \text{ if } h(\Ga)=0\end{cases} $$
for any $\Ga\in\F_q$  and $$f(P_{\infty})=\begin{cases} a_n/b_m & \text{ if } n=m \\ 0 & \text{ if } n<m\\ \infty & \text{ if } n>m\end{cases}.$$

Any automorphism $\Gs\in \Aut(K/\F_q)$ is uniquely determined by $\Gs(x)$ and given by
$$\Gs(x)=\frac{ax+b}{cx+d}$$
for some constants $a,b,c,d\in\F_q$ with $ad-bc\neq0$.
In fact, $\Aut(K/\F_q)$ is isomorphic to the projective general linear group $\PGL_2(q)$. Moreover, there exists an automorphism $\s\in  \Aut(F/\F_q)$ of order $q+1$ such that $\s$ acts cyclicly on all the rational places of $K$ \cite{HKT08,JMX20}.

\subsection{Cyclotomic function fields}
In this subsection, we briefly review the fundamental theory and results of cyclotomic function fields over finite fields.
The theory of cyclotomic function fields was developed in the language of function fields by Hayes \cite{Ha74,NX01}.

Let $q$ be an arbitrary prime power. Let $x$ be an indeterminate over $\F_q$, $R=\F_q[x]$ the polynomial ring, $F=\F_q(x)$ its quotient field, and $F^{ac}$ the algebraic closure of $F$. Let $\varphi$ be the endomorphism given by $$\varphi(z)=z^q+xz $$  for all $z\in F^{ac}$. Define a ring homomorphism
$$R\rightarrow \text{End}_{\mathbb{F}_q}(F^{ac}), f(x)\mapsto f(\varphi).$$
Then the $\F_q$-vector space of $F^{ac}$ is made into an $R$-module by introducing the following action of $R$ on $F^{ac}$, namely,
$$ z^{f(x)}=f(\varphi)(z)$$  for all $f(x)\in R$ and $z\in F^{ac}$. For a nonzero polynomial $M\in R$, we consider the set of $M$-torsion points of $F^{ac}$ defined by $$\Lambda_M=\{z\in F^{ac}| z^M=0\}.$$
In fact, $z^M$ is a separable polynomial of degree $q^d,$ where $d=\deg(M)$. The cyclotomic function field over $F$ with modulus $M$ is defined by the subfield of $F^{ac}$ generated over $F$ by all elements of $\Lambda_M$, and it is denoted by $F(\Lambda_M)$.  In particular, we list the following facts.

\begin{prop}\label{prop:2.1}
Let $P$ be a monic irreducible polynomial of degree $d$ in $R$ and let $n$ be a positive integer. Then
\begin{itemize}
\item[\rm (i)] $[K(\Lambda_{P^n}):K]=\phi(P^n)$, where $\phi(P^n)$ is the Euler function of $P^n$, i.e., $\phi(P^n)=q^{(n-1)d}(q^d-1)$.
\item[\rm (ii)] ${\rm Gal}(K(\Lambda_{P^n})/K) \cong (\mathbb{F}_q[x]/(P^{n}))^*.$ The Galois automorphism $\sigma_f$ associated to $\overline{f}\in (\mathbb{F}_q[x]/(P^{n}))^*$ is determined by $\sigma_f(\lambda)=\lambda^f$ for $\lambda\in \Lambda_{P^{n}}$.
\item[\rm (iii)]The zero place of $P$ in $K,$ also denoted by $P$, is totally ramified in  $K(\Lambda_{P^n})$ with different exponent $d_P(K(\Lambda_{P^n})/K)=n(q^d-1)q^{d(n-1)}-q^{d(n-1)}$. All other finite places of $k$ are unramified in $K(\Lambda_{P^n})/K$.
\item[\rm (iv)]The infinite place $\infty$ of $K$ splits into $\phi(P^n)/(q-1)$ places of $K(\Lambda_{P^n})$  and the ramification index $e_{\infty}(K(\Lambda_{P^n})/K)$ is equal to $q-1$. In particular, $\mathbb{F}_q$ is the full constant field of  $K(\Lambda_{P^n})$.
\item[\rm (v)] The genus of $K(\Lambda_{P^n})$ is given by $$2g(K(\Lambda_{P^n}))-2=q^{d(n-1)}\Big{[}(qdn-dn-q)\frac{q^d-1}{q-1}-d\Big{]}.$$
\end{itemize}
\end{prop}

\subsection{Kummer extensions}
The theory of Kummer extensions of function fields can be characterized as follows from \cite[Proposition 3.7.3]{St09}.
\begin{prop}\label{prop:2.2}
Let $K$ be the rational function field $\F_q(x)$ over $\F_q$ and $n$ be a positive divisor of $q-1$. Suppose that $u\in K$ is an element satisfying \[u\neq \omega^d \text{ for all } \omega\in F \text{ and } 1<d, d|n.\]
Let $F$ be the Kummer extension over $K$ defined by \[F=K(y) \text{ with } y^n=u.\] Then we have:
\begin{itemize}
\item[(a)] The polynomial $\phi(T)=T^n-u$ is the minimal polynomial of $y$ over $K$. The extension $F/K$ is Galois of degree $n$, its Galois group is cyclic, and the automorphisms of $F/K$ is given by $\sigma(y)=\zeta y$, where $\zeta\in \F_q$ is an $n$-the root of unity.
\item[(b)] Let $P\in \mathbb{P}_K$ and $Q\in F$ be an extension of $P$. Let $r_P$ be the greatest common divisor of $n$ and $\nu_P(u)$. Then
\[e(Q|P)=\frac{n}{r_P} \text{ and } d(Q|P)=\frac{n}{r_P}-1.\]
\end{itemize}
\end{prop}

\subsection{Artin-Schreier curves}
Let $q=2^n$ be a prime power for a positive integer $n$. Consider the Artin-Schreier curve defined by \[y^2+y=f(x),\] where $f(x)$ is a rational function in $\F_q(x)$.
From \cite[Proposition 3.7.8]{St09},  the properties of Artin-Schreier curves can be summarized as follows.
\begin{prop}\label{prop:2.3}
Let $K$ be the rational function field $\F_q(x)$ with characteristic two. Let $f(x)$ be a rational function in $\F_q(x)$ such that $f(x)\neq z^2+z$ for all $z\in \F_q(x)$.
Let \[F=K(y)=\F_q(x,y) \text{ with } y^2+y=f(x).\] For any place $P\in \mathbb{P}_K$, we define the integer $m_P$ by
$$
m_P=\begin{cases}
\ell &\text{ if } \exists  z\in K \text{ satisfying } \nu_P(u-(z^2+z))=-\ell<0 \text{ and } \ell \not \equiv 0 (\text{mod }2),\\
-1 & \text{ if } \nu_P(u-(z^2+z))\ge 0 \text{ for some } z\in K.
\end{cases}$$
Then we have:
\begin{itemize}
\item[(a)] $F/K$ is a cyclic Galois extension of degree two. The automorphisms of $F/K$ are given by $\sigma(y)=y+u$ with $u=0,1.$
\item[(b)] $P$ is unramified in $F/K$ if and only if $m_P=-1$.
\item[(c)] $P$ is totally ramified in $F/K$  if and only if $m_P>0$. Denote by $P^\prime$ the unique place of $F$ lying over $P$. Then the different exponent $d(P^\prime|P)$ is given by \[d(P^\prime|P)=m_P+1.\]
\item[(d)] Assume that there is at least one place $Q$ satisfying $m_Q>0$, then the full constant field of $F$ is $\F_q$ and the genus of $F$ is
\[g(F)=\frac{1}{2}\left(-2+\sum_{P\in \mathbb{P}_{F}} (m_P+1)\cdot \deg(P)\right).\]
\end{itemize}
\end{prop}

\section{Binary sequences with low correlation of length $q+1$}\label{sec:3}
Let $n$ be a positive integer. Let $q=2^n$ be a prime power and $\F_q$ be the finite field with $q$ elements.
In this section, we will construct a family of binary sequences with low correlation of length $q+1$ via cyclotomic function fields.

Let $p(x)=x^2+ax+b$ be a primitive irreducible polynomial in $\mathbb{F}_q[x]$, i.e., the least positive integer $e$ for which $p(x)$ divides $x^e-1$ is $q^2-1$.
Let $K$ be the rational function field $\F_q(x)$ over $\F_q$. Denote by $F$ the cyclotomic function field $K(\Lambda_{p(x)})$ with modulus $p(x)$ over $K$.
From \cite{Ha74,MXY16}, we have the following facts on $F$.
\begin{prop}\label{prop:3.1}
Let $p(x)=x^2+ax+b$ be an irreducible polynomial in $\mathbb{F}_q[x]$.
Let $F$ be the cyclotomic function field $K(\Lambda_{p(x)})$ with modulus $p(x)$ over $K=\F_q(x)$.
Then the following results hold:
\begin{itemize}
\item[\rm (i)] $F= K(\Gl), \text{ where } \Gl^{q^2-1} +(x^q+x+a)\Gl^{q-1} +x^2+ax+b = 0.$
\item[\rm (ii)]   There is a unique place of $F$ lying over $p(x)$ which is totally ramified in $F/K$.
\item[\rm (iii)] The infinite place $\infty$ of $K$ splits into $q+1$ rational places, each with ramification index $q-1$.
\item[\rm (iv)] All other places of $K$ except $p(x)$ and $\infty$ are unramified in $F/K$.
\item[\rm (v)]The cyclotomic function field $F$ has genus $g(F) = (q+1)(q-2)/2$.
\item[\rm (vi)] The Galois group $\Gal(F/K)\cong (\F_q[x]/(p(x)))^*$.  The automorphism $\sigma_f\in \Gal(F/K)$ associated to $\overline{f}\in (\mathbb{F}_q[x]/(p(x)))^*$ is determined by $\sigma_f(\lambda)=\lambda^f$.
\end{itemize}
\end{prop}

As $\Gal(F/K)$ is an abelian group with order $q^2-1$, there exists a normal subgroup of $\Gal(F/K)$ with order $q-1$.
Let $G\lhd \Gal(F/K)$ be the unique subgroup of $\Gal(F/K)$ with order $q-1$. Let $E$ be the fixed subfield of $F$ with respect to $G$, that is,
$$E=F^G=\{z\in F: \sigma(z)=z \text{ for any } \sigma\in G\}.$$
From Galois theory, $E/K$ is an abelian extension with degree $q+1$.
Let $Q$ be a place of $E$ lying over $p(x)$. Since $\gcd(q+1,q-1)=1$, the place $Q$ is the unique ramified place of $E$ in the extension $E/K$ from Proposition \ref{prop:3.1}. Moreover, $Q|p(x)$ is totally ramified in $E/K$ with ramification index $e(Q|p(x))=q+1$ and different exponent $d(Q|p(x))=q$.
From the Hurwitz genus formula \cite[Theorem 3.4.13]{St09}, we have
\[ 2g(E)-2=(q+1)[2g(K)-2]+2q.\]
It follows that the genus of $E$ is $0$ and $E$ is a rational function field.

Under the assumption that $p(x)$ is primitive, the item (vi) of Proposition \ref{prop:3.1} can be improved.
\begin{prop}\label{prop:3.1-2}
Let $p(x)=x^2+ax+b$ be a primitive irreducible polynomial in $\mathbb{F}_q[x]$.
Let $F$ be the cyclotomic function field $K(\Lambda_{p(x)})$ with modulus $p(x)$ over $K$.
Then the Galois group $\Gal(F/K)$ is a cyclic group of order $q^2-1$.
\end{prop}
\begin{proof}
Let $\eta$ be the $K$-automorphism of $F$ defined by $\eta(\lambda)=\lambda^x$ from Proposition \ref{prop:3.1}.
Then the automorphism $\eta^i$ is determined by $\eta^i(\lambda)=\lambda^{x^i}$ for any non-negative integer $i$.
Hence, $\eta^i=id$ (the identity automorphism) if and only if $x^i\equiv 1(\text{mod } p(x))$.
Since $p(x)$ is primitive, it is easy to see that the order of automorphism $\eta$ is $q^2-1$.
\end{proof}

Since $E$ is rational, there exist $q+1$ rational places which are lying over the infinity place $\infty$ of $K$ from Proposition \ref{prop:3.1}.
Moreover, we have $\Gal(E/K)\cong \Gal(F/K)/G$ which is a cyclic group generated by an automorphism $\sigma$ of order $q+1$ from Proposition \ref{prop:3.1-2}.
Let $P_0$ be a rational place of $E$ and $P_j=\sigma^j(P_0)$ for $1\le j\le q$. Then $P_0,P_1,\cdots, P_q$ are all the rational places of $E$ from \cite[Theorem 3.7.1]{St09}.

For any automorphism $\tau\in \Gal(E/K)$, we have $\tau(Q)=Q$ from the fact that $Q$ is totally ramified in $E/K$ from \cite[Theorem 3.8.2]{St09}.
It is clear that $\deg(Q)=2\ge 2g(E)-1=-1$. The dimension of Riemann-Roch space $\mL(Q)$ is $\ell(Q)=\deg(Q)-g(E)+1=3$.
It is easy to see that $\F_q=\mL(0)\subseteq \mL(Q)$, then there exists a vector space $V$ over $\F_q$ of dimension $2$ such that $\mL(Q)=\F_q\oplus V$.
We can define a relation $\sim$ in $V$ as follows: for any $z_1,z_2\in V$,
\[z_1\sim z_2 \Leftrightarrow \exists \tau\in \Gal(E/K) \text{ such that } z_1+\tau(z_2)\in \F_q.\]

\begin{lemma}\label{lem:3.2}
The relation $\sim$ defined as above is an equivalence relation in $V$.
\end{lemma}
\begin{proof}
It is easy to verify the relation $\sim$ satisfies the axioms of an equivalence relation:
\begin{itemize}
\item[(1)] For any $z\in V$, $z+id(z)=z+z=0\in \F_q$. Hence, we have $z\sim z$.
\item[(2)]  If $z_1\sim z_2$, then there exists $\tau\in \Gal(E/K)$ such that $z_1+\tau(z_2)=\alpha\in \F_q$.
It is easy to see that $z_2+\tau^{-1}(z_1)=\tau^{-1}(z_1)+z_2=\tau^{-1}(\Ga)=\alpha\in \F_q$.
Hence, we have $z_2\sim z_1$.
\item[(3)] If $z_1\sim z_2$ and $z_2\sim z_3$, there exist $\tau_1$ and $\tau_2$ such that $z_1+\tau_1(z_2)\in \F_q$ and $z_2+\tau_2(z_3)\in \F_q$.
It follows that $z_1+\tau_1\tau_2(z_3)=z_1+\tau_1(z_2)+\tau_1(z_2+\tau_2(z_3))\in \F_q$. Hence, we have $z_1\sim z_3$.
 \end{itemize}
\end{proof}

\begin{lemma}\label{lem:3.3}
For any element $z\in V\setminus \{0\}$ and automorphism $\tau\in \Gal(E/K)\setminus \{id\}$, one has $z+\tau(z)\notin \F_q$.
\end{lemma}
\begin{proof}
Suppose that there exist $z\in V\setminus \{0\}$ and $\tau\in \Gal(E/K)\setminus \{id\}$ such that $z+\tau(z)=\alpha\in \F_q$, then we have
$\tau(z)+\tau^2(z)=\tau(\alpha)=\Ga\in \F_q$. It follows that $z+\tau^2(z)=0$, i.e., $\tau^2(z)=z$. Let $H$ be the subgroup of $\Gal(E/K)$ generated by $\tau^2$. Then we have $z\in E^H$. Since $\tau\neq id$ and $|\Gal(E/K)|=q+1$ is odd, we have $\tau^2\neq id$. Hence, the degree of Galois extension $E/E^H$ is $[E:E^H]=|H|\ge 2$.
Let $R$ be the restriction of the place $Q\in \mathbb{P}_E$ to $E^H$. Then $Q|R$ is totally ramified in $E/E^H$ from Proposition \ref{prop:3.1}.
Since $z\in V\setminus \{0\}\subseteq \mL(Q)$, we have $-1=\nu_Q(z)=e(Q|R) \nu_R(z)=|H|\cdot \nu_R(z)<0$. Hence, we obtain $\nu_Q(z)=|H|\cdot \nu_R(z)\le -2$, i.e.,  $z\notin \mL(Q)$ which is a contradiction.
\end{proof}

For any element $z\in V$, let $[z]$ denote the equivalence class $\{x\in V: x\sim z\}$ containing $z$.
\begin{lemma}\label{lem:3.4}
For any element $z\in V$, the cardinality of each equivalence class $[z]$ is at most $q+1$.
\end{lemma}
\begin{proof}
Let $z_1$ and $z_2$ be two distinct elements in an equivalence class $[z]$. There exist automorphisms $\tau_1,\tau_2\in \Gal(E/K)$ such that
$z+\tau_1(z_1)\in \F_q$ and $z+\tau_2(z_2)\in \F_q$ from Lemma \ref{lem:3.2}. We claim that $\tau_1\neq \tau_2$. It follows that the cardinality of each equivalence class $[z]$ is at most $q+1$, since the size of $\Gal(E/K)$ is $q+1$.

Suppose that $\tau_1=\tau_2$. Then we have $z+\tau_1(z_1)+z+\tau_2(z_2)=\tau_1(z_1)+\tau_2(z_2)=\tau_1(z_1+z_2)\in \F_q$.
It follows that $z_1+z_2\in \F_q$. That is to say $z_1+z_2\in V\cap \F_q=\{0\}$. Hence, we obtain $z_1=z_2$ which is a contradiction.
\end{proof}

\begin{prop}\label{prop:3.5}
For any element $z\in V\setminus \{0\}$, there are exactly $q+1$ elements in each equivalence class $[z]$.
Moreover, there are exactly $q-1$ equivalence classes in $V\setminus \{0\}$.
\end{prop}
\begin{proof}
Let $z\in V\setminus \{0\}$. We claim that the size of the set $\{w\in V: w\sim z\}$ is $q+1$. For each automorphism $\tau\in \Gal(E/K)$, we have $\tau(z)\in \tau(\mL(Q))=\mL(\tau(Q))=\mL(Q)$ from Lemma \ref{lem:2.0}, i.e., there exists an element $\alpha\in \F_q$ such that $\tau(z)+\alpha\in V\setminus \{0\}$. Let $w_\tau=\tau(z)+\alpha$. Then we have $w_\tau+\tau(z)=\alpha\in \F_q$, i.e., $w_\tau\sim z$.

From Lemma \ref{lem:3.3}, $w_\tau$ are pairwise distinct for different $\tau \in \Gal(E/K)$. In particular,
let $\tau_1\neq \tau_2\in \Gal(E/K)$. Then there exist $w_{\tau_1}, w_{\tau_2}\in V$ such that $w_{\tau_1}+\tau_1(z)=\alpha_1\in \F_q$ and $w_{\tau_2}+\tau_2(z)=\alpha_2\in \F_q$. Suppose that $w_{\tau_1}= w_{\tau_2}$. Then we have $\tau_1(z)+\tau_2(z)=\alpha_1+\alpha_2\in \F_q$, i.e., $z+\tau_1^{-1}\tau_2(z)\in \F_q$ which contradicts Lemma \ref{lem:3.3}. The size of $V\setminus \{0\}$ is $q^2-1$, then there are exactly $q-1$ equivalence classes in $V\setminus \{0\}$.
This completes the proof.
\end{proof}

From Proposition \ref{prop:3.5}, let $[0], [z_1],[z_2],\cdots,[z_{q-1}]$ be all the pairwise distinct equivalence classes of $V$.
Let Tr be the absolute trace map from $\F_{q}$ to $\F_2$ defined by $\text{Tr}(\Ga)=\Ga+\Ga^2+\cdots+\Ga^{2^{n-1}}$ for any $\Ga\in \F_q$.
Now we choose one element from each equivalence class $[z_i]$ for $1\le i\le q-1$ to generate  a binary sequence $\bs_i$ as follows:
$$\bs_i=(s_{i,0},s_{i,1},\cdots,s_{i,q}) \text{ with } s_{i,j}=(-1)^{\text{Tr}(z_i(P_j))} \text{ with } 0\le j\le q.$$
It is clear that the length of each sequence $\bs_i$ is $q+1$ and the family of sequences $\mS=\{\bs_i: 1\le i\le q-1\}$ has size $q-1$ since these $q-1$ equivalence classes are pairwise distinct. Now we are going to show that the family of sequences  $\mS$ has low correlation.

\begin{theorem}\label{thm:3.6}
If $q=2^n$ is a prime power, then there exists a family of binary sequences $\mS=\{\bs_i: 1\le i\le q-1\}$ with correlation upper bounded by $$\text{Cor}(\mS)\le \lfloor 2\sqrt{q}\rfloor-\Gd_q,$$
where $\Gd_q$ is given by \[\Gd_q=\left\{\begin{array}{ll}0&\mbox{if $ \lfloor 2\sqrt{q}\rfloor$ is odd,}\\
1&\mbox{if $\lfloor 2\sqrt{q}\rfloor$ is even.}
\end{array}\right.\]
\end{theorem}
\begin{proof}
Firstly, let us compute the autocorrelation of the family of sequences. For each $1\le i\le q-1$, the autocorrelation of $\bs_i$ at delay $t$ with $1\le t\le q$ is given by
\begin{align*}
A_t(s_i)&=\left | \sum_{j=0}^{q} s_{i,j}s_{i,j+t}\right |=\left | \sum_{j=0}^q (-1)^{\text{Tr}(z_i(P_j))+\text{Tr}(z_i(P_{j+t}))}\right |\\ &=\left | \sum_{j=0}^q (-1)^{\text{Tr}(z_i(P_j))+\text{Tr}((\sigma^{-t}(z_i))(P_j))}\right |=\left | \sum_{j=0}^q (-1)^{\text{Tr}((z_i+\sigma^{-t}(z_i))(P_j))}\right |.
\end{align*}
The third equality follows from Lemma \ref{lem:2.0}. 
Since $z_i\in V\setminus \{0\}$ and $1\le t\le q$, we have $z_i+\sigma^{-t}(z_i)\notin \F_q$ from Lemma \ref{lem:3.3}. 
Since $\sigma(Q)=Q$, we have $z_i+\sigma^{-t}(z_i)\in \mL(Q)$ from Lemma \ref{lem:2.0}. It follows that $Q$ is the unique pole of $z_i+\sigma^{-t}(z_i)$ and $\nu_Q(z_i+\sigma^{-t}(z_i))=-1$. 
Consider the Artin-Schreier curve $E_i$ over $E$ defined by $$y^2+y=z_i+\sigma^{-t}(z_i).$$
From Proposition \ref{prop:2.3}, the place $Q$ is the unique ramified place in $E_i/E$. 
From the Hurwitz genus formula, we have $$2g(E_i)-2=2[2g(E)-2]+2\deg(Q).$$
Hence, the genus of $E_i$ is one for each $1\le i\le q-1$.
Let $N_0$ denote the number of the set $S_0=\{0\le j\le q: \text{Tr}((z_i+\sigma^{-t}(z_i))(P_{j}))=0\}$.
Let $N_1$ denote the number of the set $S_1=\{0\le j\le q: \text{Tr}((z_i+\sigma^{-t}(z_i))(P_{j}))=1\}$.
It is clear that $$N_0+N_1=q+1.$$
If $j\in S_0$, then $P_j$ splits into two rational places in $E_i$ from \cite[Theorem 2.25]{LN83} and \cite[Theorem 3.3.7]{St09}.
If $j\in S_1$, then $P_j$ is inert in $E_i$ and the place of $E_i$ lying over $P_j$ has degree two. 
Let $N(E_i)$ be the number of rational points of the Artin-Schreier curve $E_i$. Then the number of rational places of $E_i$ is 
$$N(E_i)=2N_0.$$
From the Serre bound, we have $|N(E_i)-q-1|\le \lfloor 2\sqrt{q}\rfloor.$  In particular, as $N(E_i)$ is even, we know $|N(E_i)-q-1|$ is odd and
$$|N(E_i)-q-1|\le \lfloor 2\sqrt{q}\rfloor-\Gd_q.$$
Hence, we obtain $$A_t(\bs_i)=|N_0-N_1|=|2N_0-q-1|=|N(E_i)-q-1|\le \lfloor 2\sqrt{q}\rfloor-\Gd_q.$$

Now let us consider the cross-correlation of the family of sequences. For two distinct sequences $\bs_i$ and $\bs_j$ in $\mS$, the cross-correlation of $\bs_i$ and $\bs_j$ with $1\le i\neq j\le q-1$ at delay $t$ ($0\le t\le q$) is given by
$$C_t(\bs_i,\bs_j)=\left|\sum_{k=0}^q(-1)^{\text{Tr}(z_i(P_k))+\text{Tr}(z_j(P_{k+t}))}\right|=\left|\sum_{k=0}^q(-1)^{\text{Tr}((z_i+\sigma^{-t}(z_j))(P_{k}))}\right|.$$
Since $[z_i]$ and $[z_j]$ are distinct equivalence classes, we have $z_i+\sigma^{-t}(z_j)\notin \F_q$ from Lemma \ref{lem:3.2}. 
Since $\sigma(Q)=Q$, we have $z_i+\sigma^{-t}(z_j)\in \mL(Q)$ and $\nu_Q(z_i+\sigma^{-t}(z_j))=-1$ from Lemma \ref{lem:2.0}.
Consider the Artin-Schreier curve $E_{i,j}$ defined by $$y^2+y=z_i+\sigma^{-t}(z_j).$$
From Proposition \ref{prop:2.3}, the place $Q$ is the unique ramified place in $E_i/E$ with different exponent two and the genus of $E_{i,j}$ is one.
Let $N_0^\prime$ denote the number of the set $\{0\le k\le q: \text{Tr}((z_i+\sigma^{-t}(z_j))(P_k))=0\}$.
Let $N_1^\prime$ denote the number of the set $\{0\le k\le q: \text{Tr}((z_i+\sigma^{-t}(z_j))(P_k))=1\}$.
It is clear that $$N_0^\prime+N_1^\prime=q+1.$$
Let $N(E_{i,j})$ be the number of rational points on the Artin-Schreier curve $E_{i,j}$. Similarly, we have
$$N(E_{i,j})=2N_0^\prime.$$
Again, as $N(E_{i,j})$ is even, we have   $$|N(E_{i,j})-q-1|\le \lfloor 2\sqrt{q}\rfloor-\Gd_q$$
from the Serre bound.
Hence, we obtain $$C_t(\bs_i,\bs_j)=|N_0^\prime-N_1^\prime|=|2N_0^\prime -q-1|=|N(E_{i,j})-q-1|\le \lfloor 2\sqrt{q}\rfloor-\Gd_q.$$
The proof is completed.
\end{proof}

\section{An explicit construction of sequences with low correlation}\label{sec:4}
The pervious section provides a general theoretical construction of binary sequences with low correlation via cyclotomic function fields. In fact, the sequences constructed in the above section can be realized explicitly.  In this section, we provide an explicit construction of a family of binary sequences by determining explicit automorphisms and present some numerical results.

\subsection{An explicit construction}\label{sec:4.1}
Let $p(x)=x^2+ax+b$ be a primitive irreducible polynomial in $\mathbb{F}_q[x]$. From \cite[Lemma 3.17]{LN83}, the constant $b$ should be a primitive element of $\F_q$.
Let $K$ be the rational function field $\F_q(x)$ over $\F_q$.
The cyclotomic function field $K(\Lambda_{p(x)})$ is explicitly given by $F=K(\Gl)$ with
$$\Gl^{q^2-1} +(x^q+x+a)\Gl^{q-1} +x^2+ax+b = 0.$$
Let $u=\Gl^x/x=x+\Gl^{q-1}$. Then we have $x=u+\Gl^{q-1}$ and $$ \Gl^{q^2-1}+[(u+\Gl^{q-1})^q+(u+\Gl^{q-1})+a]\Gl^{q-1}+(u+\Gl^{q-1})^2+a(u+\Gl^{q-1})+b=0.$$
It follows that 
\begin{equation}\label{eq:4.1}
\Gl^{q-1}=\frac{u^2+au+b}{u^q+u}.
\end{equation}
 Since $\Gl^{q-1}=u+x$, we have $K(\Gl)=\F_q(x,\Gl)=\F_q(u,\Gl)$.

\begin{prop}\label{prop:4.1}
Let $\eta$ be the generator of $\Gal(K(\Gl)/K)$ determined by $\eta(\Gl)=\Gl^x$ and $\tau=\eta^{q+1}$. Then we have $\langle \tau \rangle\cong \F_q^*$ and
the fixed subfield of $F$ with respect to $\langle \tau \rangle$ is $\F_q(u)$.
\end{prop}
\begin{proof}
Since the zero of $p(x)$ is totally ramified in $K(\Gl)/K$ and $\Gl$ satisfies the following equation
$$\Gl^{q^2-1} +(x^q+x+a)\Gl^{q-1} +x^2+ax+b = 0.$$
From Eisenstein's irreducibility criterion \cite[Proposition 3.1.15]{St09}, we have $x^2+ax+b$ divides $x^q+x+a$, i.e., $$x^q\equiv x+a (\text{mod } x^2+ax+b).$$
Let $\tau=\eta^{q+1}$.
It is easy to verify that $\tau(\Gl)=\eta^{q+1}(\lambda)=\lambda^{x^{q+1}}=\lambda^{x(x+a)}=\Gl^b=b\Gl$.
The third equality follows from  $x^{q+1}=x\cdot x^q\equiv x(x+a)  (\text{mod } x^2+ax+b)$ and $\Gl^{x^2+ax+b}=0$. Hence, we have $\langle \tau \rangle\cong \F_q^*$ from Proposition \ref{prop:3.1}.

It is clear that $\tau(\Gl^{q-1})=(\tau(\Gl))^{q-1}=(b\Gl)^{q-1}=\Gl^{q-1}$. Hence, we have $$\F_q(u)=\F_q(x,u)=\F_q(x,x+\Gl^{q-1})=\F_q(x,\Gl^{q-1})\subseteq F^{\langle \tau\rangle}.$$
From Galois theory, we have $[F:F^{\langle \tau \rangle}]=|\langle \tau \rangle|=q-1$.
From Equation \eqref{eq:4.1}, 
 $K(\Gl)/\F_q(u)$ can be viewed as a Kummer extension $\F_q(u,\Gl)$ over $\F_q(u)$.
Thus, we have $[F:\F_q(u)]=q-1$ from Proposition \ref{prop:2.2}.
Hence, we have $F^{\langle \tau \rangle}=\F_q(u)$.
\end{proof}

\begin{prop}\label{prop:4.2}
Let $E=\F_q(u)$ and $Q$ be the place of $E$ lying over $p(x)$. Then the place $Q\in \PP_E$ corresponds to the quadratic irreducible polynomial $p(u)=u^2+au+b$.
\end{prop}
\begin{proof}
From Proposition \ref{prop:3.1}, the zero of $p(x)$ in $\F_q(x)$ is totally ramified in $F/\F_q(x)$. Let $P$ be the unique place of $F$ lying over $p(x)$ and $Q$ be the place of $\F_q(u)$ lying over $p(x)$. Then $Q$ is the restriction of $P$ to $\F_q(u)$ and $Q$ is totally ramified in  $F/\F_q(u)$  with degree $2$.

From Proposition \ref{prop:4.1} and Proposition \ref{prop:2.2} , the zeros of $u^2+au+b$, $1/u$ and $u+\alpha$ with $\alpha\in \F_q$ are all the (totally) ramified places of $\F_q(u)$ in the extension $F/\F_q(u)$. However, the zero of $u^2+au+b$ is the unique ramified place of $\F_q(u)$ in the extension $F/\F_q(u)$ with degree $2$.
Hence, the place $Q$ is the place corresponding to the quadratic polynomial $p(u)=u^2+au+b$.
\end{proof}

The Riemann-Roch space $\mL(Q)$ is a $3$-dimensional vector space over $\F_q$, i.e.,
$$\mL(Q)=\left\{\frac{c_0+c_1u+c_2u^2}{u^2+au+b}: c_i\in \F_q \text{ for } 0\le i\le 2\right\}.$$
Hence, the complementary space $V$ of $\F_q$ in $\mL(Q)$ is a $2$-dimensional vector space over $\F_q$ spanned by $$\frac{1}{u^2+au+b}\text{ and } \frac{u}{u^2+au+b}.$$

\begin{prop}\label{prop:4.3}
Let $\sigma$ be the automorphism $\eta^{q-1}$ of $\F_q(u)$. Then the $\F_q$-automorphism $\s$ is uniquely determined by $$\sigma(u)=\frac{a^{-1}bu+b}{u+a+a^{-1}b}.$$
\end{prop}
\begin{proof}
It is easy to verify that $\sigma(u)=\s(\Gl^x/\Gl)=\Gl^{x^{q}}/\Gl^{x^{q-1}}$.
From the proof of Proposition \ref{prop:4.1}, we have $x^q\equiv x+a (\text{mod } x^2+ax+b)$ and $\Gl^{x^q}=\Gl^{x+a}.$
Assume that there exist $c,d\in \F_q$ such that $x^{q-1}\equiv cx+d  (\text{mod } x^2+ax+b)$, then we have $x+a\equiv x^q=x\cdot x^{q-1}\equiv cx^2+dx(\text{mod } x^2+ax+b)$ by multiplying $x$ at both sides.
It follows that $$\frac{c}{1}=\frac{d+1}{a}=\frac{a}{b}.$$  Hence, we have $c=ab^{-1}$, $d=a^2b^{-1}+1$ and $$x^{q-1}\equiv ab^{-1}x+a^2b^{-1}+1 (\text{mod } x^2+ax+b).$$
By the direct compution, we have
\begin{align*}
\sigma(u)&=\frac{\Gl^{x^{q}}}{\Gl^{x^{q-1}}}=\frac{\Gl^{x+a}}{\Gl^{ab^{-1}x+a^2b^{-1}+1}}=\frac{\Gl^{x}+a\Gl}{ab^{-1}\Gl^x+(a^2b^{-1}+1)\Gl}\\&=\frac{u+a}{ab^{-1}u+a^2b^{-1}+1}=\frac{a^{-1}bu+b}{u+a+a^{-1}b}.
\end{align*}
\end{proof}

From Proposition \ref{prop:4.3}, the automorphisms $\s^k\in \Gal(E/K)$ can be calculated by recursive equations for $0\le k\le q$.

\begin{prop}\label{prop:4.4}
Let $a_0=1, b_0=c_0=0,d_0=1$ and $a_1=a^{-1}b, b_1=b, c_1=1,d_1=a+a^{-1}b$.
Let $a_k,b_k,c_k,d_k$ satisfy the following recursive equations
$$\begin{cases}
a_{k+1}=a_1a_k+c_1b_k=a^{-1}ba_k+b_k,\\
b_{k+1}=b_1a_k+d_1b_k=ba_k+(a+a^{-1}b)b_k,\\
c_{k+1}=a_1c_k+c_1d_k=a^{-1}bc_k+d_k,\\
d_{k+1}=b_1c_k+d_1d_k=bc_k+(a+a^{-1}b)d_k.
\end{cases}$$
Then the automorphism $\s^k$ of $\F_q(u)$ can be  determined by
$$\s^k(u)=\frac{a_ku+b_k}{c_ku+d_k}.$$
\end{prop}
\begin{proof}
Assume that $\s^k$ is the automorphism of $\F_q(u)$ determined by
$$\s^k(u)=\frac{a_ku+b_k}{c_ku+d_k}.$$
Then we have
$$\s^{k+1}(u)=\s(\s^k(u))=\s\left(\frac{a_ku+b_k}{c_ku+d_k}\right)=\frac{a_k\s(u)+b_k}{c_k\s(u)+d_k}=\frac{a_k\frac{a_1u+b_1}{c_1u+d_1}+b_k}{c_k\frac{a_1u+b_1}{c_1u+d_1}+d_k}.$$
From the above recursive equations, we obtain
$$\s^{k+1}(u)=\frac{(a_1a_k+c_1b_k)u+(b_1a_k+d_1b_k)}{(a_1c_k+c_1d_k)u+(b_1c_k+d_1d_k)}   =\frac{a_{k+1}u+b_{k+1}}{c_{k+1}u+d_{k+1}}.$$
This completes the proof.
\end{proof}

Let $P_0$ be the zero of $u$ in $E=\F_q(u)$. From Proposition \ref{prop:4.4}, if $a_j\neq 0$, then $P_j=\s^j(P_0)$ corresponds to the linear polynomial $u+a_j^{-1}b_j$; otherwise, $P_j=P_\infty$ which is the infinity place of $\F_q(u)$.
For any $z\in V\setminus \{0\}$, we have $$[z]=\cup_{i=0}^q ((\s^i(z)+\F_q) \cap V)$$ from Proposition \ref{prop:3.5} and Proposition \ref{prop:4.4}.
Thus, we can obtain all the equivalence classes $[z_1],[z_2],\cdots,[z_{q-1}]$ of $V\setminus \{0\}$ under the relation $\sim$.
Furthermore, $z_i(P_j)=z_i(u)|_{u=-b_j/a_j}$ if $a_j\neq 0$; otherwise, $z_i(P_j)=z_i(P_\infty)=0$.
Hence, we can provide an explicit construction of a family of binary sequences with low correlation via cyclotomic function fields from Theorem \ref{thm:3.6}.

\subsection{Algorithm and numerical results}\label{sec:4.2}
Now we provide an algorithm to generate a family of sequences with low correlation.
Such an explicit construction of sequences with low correlation can be given as follows.

\bigskip
\begin{center}
{\bf Construction of sequences with low correlation}
\end{center}
\begin{itemize}
\item {\bf Step 1:} Input  a prime power $q=2^n$.
\item {\bf Step 2:} Find a primitive quadratic irreducible polynomial $p(x)=x^2+ax+b\in \F_q[x]$.
\item {\bf Step 3:} Define $\s$ as an $\F_q$-automorphism of $\F_q(u)$ determined by
$$\sigma(u)=\frac{a^{-1}bu+b}{u+a+a^{-1}b}.$$
\item {\bf Step 4:} Let $a_0=d_0=1,b_0=c_0=0$ and $a_1=a^{-1}b, b_1=b, c_1=1,d_1=a+a^{-1}b$. For $0\le k\le q$, calculate the explicit expressions of $\s^k(u)=\frac{a_ku+b_k}{c_ku+d_k}$ via the following recursive equations:
$$\begin{cases}
a_{k+1}=a_1a_k+c_1b_k=a^{-1}ba_k+b_k,\\
b_{k+1}=b_1a_k+d_1b_k=ba_k+(a+a^{-1}b)b_k,\\
c_{k+1}=a_1c_k+c_1d_k=a^{-1}bc_k+d_k,\\
d_{k+1}=b_1c_k+d_1d_k=bc_k+(a+a^{-1}b)d_k.
\end{cases}$$
\item {\bf Step 5:} Let $P_0$ be the zero of $u$ in $\F_q(u)$ and $P_j=\s^j(P_0)$ for $1\le j\le q$.
\item {\bf Step 6:} Let $V$ be the $2$-dimensional vector space over $\F_q$ spanned by $$\frac{1}{u^2+au+b} \text{ and } \frac{u}{u^2+au+b}.$$
\item {\bf Step 7:} Determine the equivalence classes of $V\setminus \{0\}$ under the relation $\sim$.
Choose an element $z_1=\frac{a_{1,0}+a_{1,1}u}{u^2+au+b} \in V\setminus \{0\}$, for any automorphism $\s^i$ with $0\le i\le q$, $\s^i(z_1)$  can be explicitly given by
$$\s^i(z_1)=\frac{a_{1,0}+a_{1,1}\s^i(u)}{(\s^i(u))^2+a\s^i(u)+b}.$$
Determine the intersection
$$(\s^i(z_1)+\F_q) \cap V=\{w_{1,i}\}. $$
The equivalence class of $z_1$ is given by $[z_1]=\{w_{1,i}: 0\le i\le q\}.$
Choose $z_2\in V\setminus \{0\}-[z_1]$, similarly we can obtain the equivalence class of $z_2$.
Recursively, we can obtain all the equivalence classes $[z_1],[z_2],\cdots,[z_{q-1}]$ for $V\setminus \{0\}$.
\item {\bf Step 8:} Output a family of sequences $\mS=\{s_i: 1\le i\le q-1\}$ defined by $$s_i=(s_{i,0},s_{i,1},\cdots,s_{i,q}) \text{ with } s_{i,j}=(-1)^{\text{Tr}(z_i(P_j))} \text{ for } 0\le j\le q.$$
\item {\bf Step 9:} 
Output the correlation $Cor(\mS)=\max\{\{|\sum_{k=0}^{q} s_{i,k}s_{i,k+t}|:1\le i\le q-1, 1\le t\le q\} \cup \{|\sum_{k=0}^{q} s_{i,k}s_{j,k+t}|: 1\le i\neq j\le q-1, 0\le t\le q\}\}.$
\end{itemize}

We list some numerical results on the family of binary sequences obtained from the above construction in Table II.


\begin{table}[]\label{tab:2}
	\setlength{\abovecaptionskip}{0pt}%
	\setlength{\belowcaptionskip}{10pt}%
	\caption{PARAMETERS OF OUR SEQUENCES}
	\center
	\begin{tabular}{@{}|c|c|c|c|c|@{}}
		\toprule
		Field Size & Seq. Length & Family Size & Max Correlation & No. of Balanced Seq. \\ \midrule
		32    & 33             & 31       &    11                    &   5         \\ \midrule
		64    & 65              & 63         & 15                       & 12            \\ \midrule
		128   & 129            & 127        & 21                        & 14           \\ \midrule
		256  & 257             & 255        & 31                        & 16           \\ \midrule
		512   & 513             & 511         &  45                      & 18            \\
		\bottomrule
	\end{tabular}
\end{table}

\end{document}